\documentclass[letterpaper, 10 pt, conference]{ieeeconf}  

\IEEEoverridecommandlockouts                              

\usepackage{amsmath}
\usepackage{amsfonts}
\usepackage[pdftex]{graphicx}
\usepackage{float}
\usepackage{algorithm}
\usepackage[noend]{algpseudocode}

\usepackage{amsthm}

\theoremstyle{definition}
\newtheorem{definition}{Definition}
\theoremstyle{plain}

\newtheorem{proposition}{Proposition}
\usepackage{accents}

\usepackage[colorinlistoftodos]{todonotes}

\title{\LARGE \bf
Between-Ride Routing for Private Transportation Services
}

\author{Ian Schneider$^{1}$, Jun Jie Joseph Kuan$^{2}$, Mardavij Roozbehani$^{3}$, Munther Dahleh$^{3}$
\thanks{$^{1}$Institute for Data, Systems, and Society, Massachusetts Institute of Technology (MIT)
        {\tt\small ischneid@mit.edu}}%
\thanks{$^{2}$Department of Electrical Engineering \& Computer Science, MIT}%
\thanks{$^{3}$Laboratory for Information and Decision Systems, MIT}%
}

\begin{document}

\maketitle
\thispagestyle{empty}
\pagestyle{empty}

\begin{abstract}

Spurred by the growth of transportation network companies and increasing data capabilities, vehicle routing and ride-matching algorithms can improve the efficiency of private transportation services. However, existing routing solutions do not address where drivers should travel after dropping off a passenger and before receiving the next passenger ride request, i.e., during the between-ride period. We address this problem by developing an efficient algorithm to find the optimal policy for drivers between rides in order to maximize driver profits. We model the road network as a graph, and we show that the between-ride routing problem is equivalent to a stochastic shortest path problem, an infinite dynamic program with no discounting. We prove under reasonable assumptions that an optimal routing policy exists that avoids cycles; policies of this type can be efficiently found. We present an iterative approach to find an optimal routing policy.
Our approach can account for various factors, including the frequency of passenger ride requests at different locations, traffic conditions, and surge pricing. We demonstrate the effectiveness of the approach by implementing it on road network data from Boston and New York City. 

\end{abstract}

\section{INTRODUCTION}

Advances in information technology and decision theory are helping revolutionize the market for private transportation services. Transportation Network Companies (TNCs) like Lyft and Uber utilized internet-enabled ride requests to quickly grow their market share in private transportation. From 2014 to 2016, Uber's market share for ride-hailing services rose from 18\% to 75\% in the United States 
 \cite{Martha:2016}. Policy efforts increasingly support the concept of mobility as a service \cite{jittrapirom2017mobility}, increasing the economic importance of private transportation services. However, TNCs are also associated with concerns related to increased congestion and low driver wages.  

Tools from decision theory can utilize new sources of consumer and driver data to help improve the efficiency of private transportation services. TNCs typically feature automated passenger ride matching and prices that vary geographically and temporally. Higher prices during peak demand periods are commonly referred to as surge prices. There is new opportunity for developing efficient algorithms to optimize surge pricing, improve system efficiency, and to help increase driver wages. 

In this paper, we consider the situation where a driver seeks to maximize the expected value of their profits by optimizing routing decisions during the between-ride period, the time period after a driver drops off a passenger and before they receive their next ride request.

We formulate this problem as an undiscounted dynamic program with an uncertain number of decision stages; our formulation is equivalent to a stochastic shortest path problem. We prove that the between-ride problem has special characteristics; given reasonable assumptions, it can be efficiently solved as a dynamic program with a finite number of stages.  
We implement our algorithm using road network data from Boston and New York City, demonstrating that our approach is practical and scalable. Our implemented algorithm can directly advise drivers between rides, which could bring substantial benefits to transportation efficiency. It can also help increase driver wages and reduce costs.

Current research approaches regarding algorithms for private transportation services focus on the issue of driver-passenger matching.
They encompass a wide variety of scenarios, including but not limited to: matching algorithms for ride-sharing and carpooling services \cite{agatz2012optimization} \cite{SCHREIECK2016272} \cite{kleiner2011mechanism}, private transportation services with ride-sharing \cite{alonso2017demand}, and generalizations of the ride-matching problem to cases where passengers might be asked to transfer over to another vehicle in the middle of their journey \cite{masoud2017decomposition}. 
Algorithms for driver-passenger matching typically focus on the period after passengers have made a trip request. Dispatch algorithms can move drivers to higher value areas, but they do not provide routing suggestions to private drivers to optimize their personal profits. 

Our approach complements the aforementioned work on ride-matching: our algorithm can be used to navigate drivers towards areas with high probability of new passenger requests by providing routing directions that align with their profit motive. These results can improve the effectiveness of existing algorithms by helping move unmatched drivers to higher value locations where they are more likely to be matched to a ride when a ride request is issued.

The issue of location-based and distance-based pricing policy is also an active area of research. Research demonstrates how spatial pricing policy can be implemented to achieve better matching between supply of drivers and demand of passengers \cite{bimpikis2016spatial}. 
Another study investigates the differences between dynamic and static pricing strategies \cite{banerjee2015pricing}. Each of these papers focus on pricing strategies, and drivers are assumed to move efficiently towards areas with higher demand and prices. It is unlikely that drivers are able to make optimal between-ride decisions without the assistance of routing technology; the between-ride routing problem is complicated, requiring the synthesis of multiple data sources. 

Again, our proposed algorithm complements existing research: it provides optimal paths for drivers between rides to maximize driver profits, given an existing pricing policy. If a pricing policy is efficient, our approach can ensure that the benefits of the pricing policy are attained. The combination of research areas can help create a more coordinated network system, reducing waiting time for passengers and allowing more driver-passenger matches to occur.


Section \ref{model} describes the model and formulates the between-ride routing problem as a stochastic shortest path problem, a type of dynamic program. Section \ref{results} proves that we can solve the between-ride routing problem by using a finite-horizon dynamic program, where the number of stages is no greater than the number of nodes in the transportation network. Section \ref{sec:alg} presents a practical algorithm for finding an optimal solution to the between-ride routing problem, and Section \ref{sec:implementation} describes an implementation of our approach using road network data from Boston and New York City. 

\section{Technical Model} \label{model}

This section details the model for optimizing a between-ride route in order to maximize the expected value of profits for a driver. First, we detail the relevant parameters. Second, we explain the probability model whereby a driver receives ride requests at a particular location according to an exponential distribution. We specify the expected value of a route along points $x \in R^2$ on a 2-dimensional map. In the following Subsection \ref{model2}, we present a model of the road network as a directed, connected graph. We take the parameters of interest to be constant along each edge. We explain the driver objective function and optimization problem in terms of discrete decisions on the graph.  

Consider a specific driver that does not currently have a passenger. Let $w$ be the driver's wage rate, i.e. the value of their time. Let $f$ be their fuel and vehicle cost per unit distance driven. 

For each location $x$ on the map, let $R_x$ be the expected value of the profit from a ride request the driver receives at location $x$. We assume that the random profit from the ride accounts for its various features, including length, price, and duration. For tractability, we assume that rides are undifferentiated aside from the expected value of their profit. In practice, ride opportunities can vary in other ways as well. For instance, rides at some locations could be more likely to end at high-value locations, which would increase the expected value of profit from subsequent rides. In practice, this could be incorporated into these results by adjusting $R_x$ to account for relevant characteristics, but a proper formulation would be non-myopic with regard to the value of subsequent rides. Future research could focus on the case where the value of subsequent rides is directly incorporated to the model; this would lead to an interesting formulation over multiple potential rides. 

Furthermore, let $Q_x$ be the the pickup rate at location $x$, i.e. the expected value of the number of ride requests per minute. The values $R_x$ and $Q_x$ are indexed by $x$ because they can vary at different points in the transportation network. These variables can change over time, but the algorithm assumes that they are static over the course of the between-ride decision making. This is a reasonable assumption given the short duration of between-ride routing; research from Denver suggests that the average between-ride period is less than 12 minutes, with a median of 7.5 minutes \cite{henao2017impacts}. If inputs change or shocks occur, the algorithm can be rerun to optimize the remaining route during the between-ride period. This allows drivers to respond to real-time changes in demand and congestion.

Our model assumes that $R_x$ and $Q_x$ are not influenced by the driver's route. It would also be interesting to consider the case where drivers actions directly influence the price and demand for rides, for instance in a game-theoretic or mean-field model.  

At any position $x$, we assume that ride requests arrive according to a Poisson process, and we model passenger ride requests via an exponential distribution. Consider a driver at position $x(0)$ who travels along the trajectory $x(t)$ in continuous time, and let $M$ be the random variable of the time the driver receives their first ride request along route $x(t)$. Then
\begin{equation}
\label{eqn:2}
\mathbb{P}(M \geq t) = \exp⁡ \Big(\int_{0}^{t} -Q_{x(s)}\ ds\Big)
\end{equation}

In general, we assume the drivers must accept their first ride request, as is commonly required for drivers in TNCs. Let $J$ be the expected value of the profit through the next ride period for the between-ride driver on trajectory $x(t)$ at each time $t$. $J$ is given by
\begin{equation}
\label{eqn:1}
J = \int_{0}^{\infty} \Big( R_{x(t)}  Q_{x(t)} -w-f \frac{dx}{dt} \Big) \mathbb{P}(M \geq t) \ dt
\end{equation}

In Equation \ref{eqn:1}, the first term $R_{x(t)}  Q_{x(t)}$  corresponds to the expected revenue earned at $x(t)$, taking into account the likelihood of receiving a passenger match while at $x(t)$. The second term $w$ is the cost due to time spent waiting, while the third term $f \frac{dx}{dt}$ represents fuel costs. The trajectory $x(t)$ is differentiable, so $\frac{d x(t)}{d t}$ exists and is finite for all $t$. It represents the speed traveled along the route, so it accounts for local vehicle speeds and congestion. 

From Equations \eqref{eqn:2} and \eqref{eqn:1} we obtain the value of a driver waiting stationary at a particular location $z$, i.e. with $x(t) = z$ for all $t$:

\begin{equation}
\label{eqn:3}
\tilde{J} = \int_{0}^{\infty} (R_z Q_z-w) e^{-t Q_z }\ dt = R_z-\frac{w}{Q_z}
\end{equation}
This expression is intuitive: if a driver waits at location $z$ until they receive a passenger match, they receive a ride match eventually with probability $1$. Thus, their expected revenue is $R_z$. The cost of waiting is their wage rate $w$ times the expected value of the amount of time until they receive a passenger match $Q_{z}^{-1}$.

In general, a driver at location $x(0)$ seeks to choose a trajectory $x(t)$ to maximize \eqref{eqn:1}. In the following section, we will model the road network as a graph and formulate the decision problem as a dynamic program in discrete time and space. As we will show, this leads to a tractable decision problem that can be efficiently solved. 

\setcounter{figure}{0}
\begin{figure}
\begin{center}
\includegraphics[width=0.75\columnwidth]{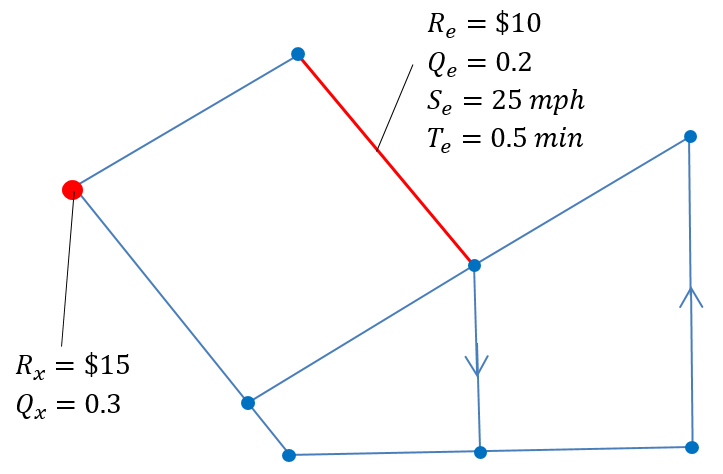}
\caption{Model of a simple road network.}
\label{fig:1}
\end{center}
\end{figure}

\subsection{Road Network Model} \label{model2}
To solve the between-ride routing problem, we model the road network as a directed, connected graph $G = (N, E)$. Each edge $e \in E$ represents a section of a road, while the set of nodes $N$ includes, but is not limited to, all road intersections.\footnote{Nodes can also be used to model a specified point along a road segment where there is no intersection with other roads. These additional nodes are equivalent to intersections with only 1-2 options for directed travel: continue straight, or (possibly) make a u-turn.} Let $|N| = n$. Note that $E$ includes all loops; i.e. $\forall i \in N$, $(i,i) \in E$. A driver on the edge $(i,i)$ corresponds to the action of a driver waiting at node $i$. 

The driver seeks to choose a route to maximize their expected profit over an infinite horizon:
\begin{equation} \label{DP}
	\max\limits_{\mu_k(x_k)} \lim\limits_{N \rightarrow \infty} \mathbb{E} [\sum\limits_{k=0}^{N - 1} g(x_k, x_{k+1}, \mu_k(x_k))]
\end{equation}

We can formulate the between-ride routing problem \eqref{DP} as a stochastic shortest path problem\footnote{Stochastic shortest path problems were first formulated by \cite{eaton1962optimal}. Existing research \cite{bertsekas1991analysis} extends the analysis to the case where transition values may be positive or negative, which is helpful for our analysis. The two-volume textbook \cite{Bertv1}\cite{Bertv2} provides additional information on stochastic shortest path problems in each volume.}, which is an infinite horizon dynamic program and a type of Markovian Decision Problem (MDP). Let the state space $X = N \cup \{m\}$, where we augment the set of nodes $N$ with a terminal state $m$. When a driver moves to state $m$, we say that they are no longer in the between-ride routing period, because they have either found a rider or stopped offering rides. Once the driver reaches state $m$, they derive zero additional cost or value, so the decision problem equivalently terminates.

The function $g(x_k, x_{k+1}, \mu_k(x_k))$ represents the value of transitioning to state $x_{k+1}$ from state $x_k$ after choosing policy $\mu_k$, which prescribes an action for each state in each stage of the decision problem. It includes the potential value of receiving a ride match, less the costs associated with the driver's time and fuel costs. The transition probabilities are stationary, and they are given by $p_{ij}(\mu) = \mathbb{P}(x_{k+1} = j | x_k = i, \mu_k = \mu )$. 

At a specific state $i \in X$, the driver chooses the next location $j$ with edge $(i,j) \in E$. For example, at node $i \in N$, the driver could choose to stay at node $i$, turn right to travel to node $k$, or turn left to travel to node $\ell$ if $(i,i), (i,k), (i,\ell) \in E$. If $j \neq i$, then the time to traverse the edge from $i$ to $j$ is not a decision variable; it is given by the speed of traffic. If $j = i$, then we say that the driver is waiting at node $i$, and the driver can choose exactly how long to wait at node $i$ before making a subsequent decision. 

Formally, the driver chooses an action $u(i) \in U(i)$ where $U(i)$ is the set of admissible actions at state $i$. The action $u(i)$ is a double, i.e. $u(i) = (u_1, u_2)$. The driving decision $u_1 \in \{j \in N | (i,j) \in E\} \cup \{m\}$ is the driver's choice of the next location. The waiting or travel time $u_2$ is chosen from the set $u_2 \in (0,\infty]$ if $u_1 = i$. Otherwise, $u_2 = T_{(i,j)}$. $T_{(i,j)}$ is the time required to travel from node $i$ to $j$ along edge $(i,j)$ at the current speed of traffic, taking into account traffic and congestion.

Additionally, for all $i \in N$, there is an admissible action $u(i)$ with $u_1 = m$. This describes the case where the driver stops searching for another ride request and stops offering rides. This can happen, for instance, at the the end of a shift or if prices are too low for the driver to keep searching for the next ride request. 


Consider the transition probabilities and profits for a driver at node $i \in N$. Let the driver elect action $u(i) = (j,t) \in U(i)$ with $e = (i,j) \in E$. If $j = i$, then $t$ is a decision variable; otherwise, $t = T_e$. Along edge $e$, the constant $Q_e > 0 $ represents the arrival rate for ride requests. Then, from \eqref{eqn:2} 

\begin{equation} \label{MDP_prob}
p_{ik}(u) = 
	\begin{cases}
	1-e^{- t Q_e} & k = m \\
	e^{- t Q_e} & k = j \\
	0 & \text{otherwise} 
	\end{cases}
\end{equation}


Remember that if $j \neq i$, then $t = T_e$, because drivers must move at the speed of traffic along a particular edge.\footnote{The original map can be augmented with a node in the middle of edges with curbside parking or waiting zones, in order to model the case where drivers can stop and wait along some edges.} In this case, there is a probability $1-e^{- Q_{e} T_{e}}$ that they will be matched with a ride; otherwise, they will move to node $j$ for the next decision stage after a duration $T_e$.  

A driver making the decision $u(i) = (j,t)$ at node $i$ receives profit
\begin{equation} \label{g}
g(i,k,u) = 
\begin{cases}
	R_e - \phi_e(t) (w + f S_e) & k = m \\
	- t (w + f S_e) & k = j \\
	0 & \text{otherwise} 
	\end{cases}
\end{equation}
As before, the values $w$ and $f$ refer to the driver's fixed wage rate and fuel/vehicle cost. The expected value of the ride revenue for every edge $e \in E$ is known and stationary in the time period of interest, and given by $R_e > 0 $. The driver drives at a constant speed $S_e$ along each edge (which is directly the edge distance divided by $T_e$). 

The variable $\phi_e(t)$ is the expected value of the time until a match occurs along edge $e$, when the edge is traversed over a time period $t$, conditional on the fact that a match does occur in $[0,t]$; $\phi_e(t) = \mathbb{E}[M_e | M_e < t]= Q_e^{-1} + t(1 - e^{Q_e t})^{-1}$.

%

From \eqref{g}, we can write the expected value of the profit at any stage $i$ according to the transition probabilities and duration associated with the chosen action $u(i) = (j,t)$, again with $e = (i,j)$: 
\begin{equation} \label{exprev}
	\begin{aligned}
		g (i,u) & = \mathbb{E} [g(i,k,u)]  \\
					  & = (1 - e^{-t Q_e}) (R_e - \frac{w + f S_e}{Q_e})
	\end{aligned}
\end{equation}
Note that $S_h = 0$ for loops $h$, i.e. edges of the form $h =(i,i)$. The idea is that the only cost for drivers when they are waiting is due to their time, not due to gas or other per-distance vehicle costs. When $u_2 = \infty$, we evaluate the associated transition probabilities and profits as the limits of the provided equations when $t$ goes to infinity; these limits exist for each of the provided expressions in \eqref{MDP_prob}, \eqref{g}, and \eqref{exprev}. 





Since $m$ is the terminal node, $p_{m m}(u) = 1$ and $g(m,t)(u) = 0$ for any $u \in U(m)$. The driver receives value when they transition to $m$ randomly by receiving a ride request. Once a driver reaches $m$, they have accepted a ride request or stopped offering rides. They receive no additional value and the decision problem ends. As explained previously, at any node $i \in N$, the driver can select to stop offering rides; mechanistically, this is performed as an action to move directly to the terminal node without any reward. For $u(i)$ with $u_1 = m$, $p_{im}(u) = 1$ and $g(i,u) = 0$. The idea is that the driver can elect to stop offering rides at any time, which incurs no further cost but eliminates the opportunity of collecting revenue from a potential ride. 

Going forward, we assume that local maxima of $R_e - \frac{w}{Q_e}$ are defined as nodes. Consider an edge $e = (j,k)$ with $j \neq k$. Then 
\begin{equation} \label{assume1}
	R_e - \frac{w}{Q_e} \leq \max\limits_{i = \{j,k\}} \{R_{(i,i)} - \frac{w}{Q_{(i,i)}}\}
\end{equation}
From any input data, it is straightforward to ensure that a map meets the required assumption by adding a node along or in the middle of any edge that has value $R_e - \frac{w}{Q_e}$ greater than that value for each of their adjacent nodes.  

The policy $\mu$ defines actions $u(i) \in U(i)$ for each state $i \in X$ in each decision period $k$.  In practice, the policy $\mu = (\mu_1,\mu_2,...)$ can vary in each decision stage. Therefore, the action taken at node $i$ in stage $k$ under policy $\mu$ is $\mu_k(i) = (\mu_{k,1}(i), \mu_{k,2}(i))$. However, we focus on stationary policies where $\mu = (\mu, \mu, ...)$ at each stage. As we will show, at least one of these stationary policies is optimal, justifying our narrow focus. Due to the focus on stationary policies, in the subsequent Section we drop the stage $k$ subscript from our notation. Instead, $\mu$ refers to stationary policies and $\mu_1(i)$ and $\mu_2(i)$ refer to the actions (direction and waiting time) associated with policy $\mu$ at node $i \in N$.

\setcounter{figure}{1}
\begin{figure*}[t]
\begin{center}
\includegraphics[width=1.8\columnwidth]{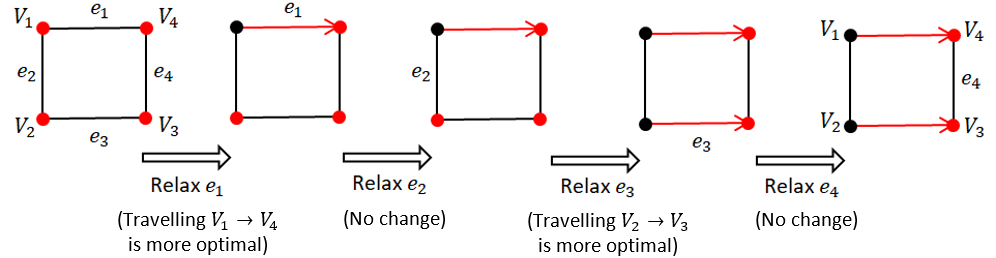}
\caption{First few steps of our algorithm on a simple road network. Red arrows indicate optimal paths found so far. Initially, the algorithm sets all drivers to stay where they are. By the end of the fourth step, the algorithm dictates that drivers at $V_1$ and $V_2$ should go to $V_4$ and $V_3$ respectively, while drivers at $V_4$ and $V_3$ stay where they are.}
\label{fig:2}
\end{center}
\end{figure*}

\section{Results} \label{results}

Optimal policies for infinite horizon dynamic programs like \eqref{DP} are typically solved using convergence of value iteration (VI) algorithms \cite{Bertv1}. In general, this can lead to sub-optimal results despite extensive computation periods; this could limit their value for the between-ride routing problem, since driving suggestions should ideally be provided to multiple drivers in a network and available very quickly after a driver drops off their previous passenger and enters the between-ride period. However, in this section, we show that the between-ride routing problem has special structure that guarantees that allows it to be efficiently solved as a finite-horizon dynamic program with $n$ stages. 

Since every node and edge has a non-zero probability of a match,\footnote{We could relax this assumptions, but in practice a match is possible anywhere, because drivers can be matched with passengers that are near their location. A driver can receive a match in a nearby residential neighborhood, for instance, while traveling along a highway.} then for any policy, there is always a positive probability that the termination state will be reached. 

Then as explained by \cite{Bertv1}, optimal values $J^*(1),...J^*(n)$ at nodes $i \in N$ indexed by $1,...n$ satisfy Bellman's Equation:
\begin{equation} \label{Bellman}
	J^*(i) = \min\limits_{u \in U(i)} \{g(i,u) + \sum_{j=1}^{n} p_{ij}(u) J^*(j)\} \quad i = 1,...n
\end{equation}

First, we will show that there is exists optimal stationary policy where for all states $i$, $\mu(i)$ obtains the minimum of \eqref{Bellman} and where the driver waits until they receive a ride request at any node where they choose to wait. Then we will show that this optimal stationary policy has no cycles, i.e. $\mathbb{P}(x_k = i, x_{\ell} = i) = 0$ for any states $i \in N$, and stages $k, \ell \in \mathbb{N}$. Together, this implies that an optimal policy can be found using a type of value iteration algorithm requiring at most $n+1$ steps. 

\begin{proposition} \label{stationarypolicy}
	There exists a stationary policy $\mu^*$ that is optimal for the decision problem \eqref{DP}. Under this stationary policy, whenever a driver waits at a specific node, they plan to remain at that node until they receive a ride request, i.e. $\forall i \in N$,  $\mu^*_1(i) = i$ implies $\mu^*_2(i) = \infty$. 
\end{proposition}

\begin{proof}
	First, let $\tilde{U}(i)$ define a restricted action space for node $i$. We say that $u(i) = (u_1, u_2) \in \tilde{U}(i)$ if $u_1 \in \{j \in N | (i,j) \in E\} \cup \{m\}$ and $u_2 = \infty$ if $u_1 = i$ and $u_2 = T_{(i,j)}$ if $u_2 = j \neq i$. Note that $\forall i \in N$, $\tilde{U}(i) \subset U(i)$. This restricted action space implies that whenever a driver chooses to wait at a node, they will wait at that node until they receive a ride request.
	
	The values of $J^*(1),...J^*(n)$ given by $\eqref{Bellman}$ are unchanged if we replace the constraint sets $U(i)$ with $\tilde{U}(i)$. To see this, fix $J^*(1),...J^*(n)$ according to \eqref{Bellman}, and for each $i$ fix a set of actions $\mu^*(i)$ that attain the optimal values of \eqref{Bellman} for each state $i$. Define the set $I = \{i | \mu^*_1(i) = i\}$. Define an additional policy $\mu$ such that $\forall i \notin I$, $\mu(i) = \mu^*(i)$, and $\forall i \in I$, $\mu_1(i) = \mu^*_1(i) = i$ but  $\mu_2(i) = \infty$. 
	
	For an arbitrary policy $\tilde{\mu}$, let there be some $i\in N$ with $\tilde{\mu}_1(i) = i$ and $\tilde{\mu}_2(i) = t_i$. Let $\ell = (i,i)$. Then $\forall i$ with $\tilde{\mu}_1(i) = i$, 
	\begin{equation} \label{lemma_eq1}
	J_{\tilde{\mu}}(i) =  (1 - e^{-t_i Q_\ell}) (R_\ell - \frac{w}{Q_\ell}) + e^{- t_i Q_\ell} J_{\tilde{\mu}}(i) 
	\end{equation}
	due to equations \eqref{MDP_prob} and \eqref{exprev} and by definition of the terminal state with $J(m) = 0$. This implies that $\forall i \in I$
	\begin{equation} \label{lemma_eq2}
	J_{\mu^*}(i) = R_{\ell} - \frac{w}{Q_{\ell}} = J_\mu(i)
	\end{equation}
	The first equality is found by evaluating \eqref{lemma_eq1} for $\tilde{\mu}(i) = \mu^*(i)$ and algebraically solving for $J_{\mu^*}(i)$. The second equality is found by evaluating \eqref{lemma_eq1} for $\mu(i)$ by taking the limit of the right hand side as $t_i \rightarrow \infty$.
	
	Subsequently, this implies that $\forall i \notin I$, $J_{\mu^*}(i) = J_\mu(i)$, since for each $i \notin I$, the objective function in $\eqref{Bellman}$ is unchanged. Therefore, we can restrict the decision space to $\tilde{U}_i$, and the same $J^*(1),...J^*(n)$ as above satisfy
	\begin{equation} \label{Bellman2}
	J^*(i) = \min\limits_{u \in \tilde{U}(i)} \{g(i,u) + \sum_{j=1}^{n} p_{ij}(u) J^*(j)\} \quad i = 1,...n
	\end{equation} 
	Therefore, any policy $\mu$ that obtains the minimum in \eqref{Bellman2} also obtains the minimum in \eqref{Bellman}. 
	
	For each state $i$, the optimization problem in \eqref{Bellman2} describes a finite and discrete choice set with $|\tilde{U}(i)| \leq n + 1$. Therefore, there exists some $\mu$ such that for all $i \in N$, $\mu(i)$ obtains the maximum in \eqref{Bellman2}. This $\mu$ describes a stationary policy, because the objective function and constraints in \eqref{Bellman2} do not change across decision periods. We showed that $\mu$ also obtains the maximum of \eqref{Bellman}. As shown by \cite{Bertv1}, this implies that $\mu$ is an optimal policy for the decision problem \eqref{DP}. \qed
\end{proof}	
\begin{proposition} \label{simplepath2}
	There exists an optimal stationary policy $\mu$ that meets the characteristics of Proposition \ref{stationarypolicy} and has no cycles with probability 1, i.e. $\mathbb{P}(x_k = i, x_{k'} = i) = 0$ for any state $i \in N$ and decision stages $k,k' \in \mathbb{N}$.    
\end{proposition}
\begin{proof}
	Let $\mu$ refer to an optimal stationary policy that meets the criteria of Proposition \ref{stationarypolicy}. By examining \eqref{MDP_prob}, we see that $\mathbb{P}(x_k = i, x_{k'} = i) > 0$ under policy $\mu$ only if the policy admits a cycle (possibly of length 1): $\exists \{\ell_1, \ell_2, ..., \ell_p\} \subset N$ such that under $\mu$, we have that $\mu_1(\ell_1) = \ell_2$, $\mu_1(\ell_2) = \ell_3$, ..., $\mu_1(\ell_p) = \ell_1$. For each $i \in 1,...p$, let $e_i = (\ell_i, \mu_1(\ell_i)) \in E$. 
	
	Assume the policy admits a cycle. Without loss of generality, let $\ell_1$ be the maximum valued waiting node in the loop: 
	\begin{equation} \label{bestnode}
	    \ell_1 = \arg\max\limits_{i \in \{\ell_1, \ell_2, ..., \ell_p\}} \{R_{(i,i)} - \frac{w}{Q_{(i,i)}}\} 
	\end{equation}
	and let $V(\ell_1)$ be the optimal value of the objective function in \eqref{bestnode}.
	Then, by \eqref{assume1} and because $f, S_e > 0$, for all $i \in 1,....p$ we have that $g(\ell_i,\mu) \leq p_{\ell_i m}(\mu) V(\ell_1)$.
	
	Then, we can write Bellman's equation for the policy $\mu$ as  
	\begin{multline} \label{lemma_eq1_a}
	J_{\mu}(\ell_1) =  g(\ell_1, mu) + p_{\ell_1 \ell_2} (\mu) g(\ell_2, mu) \\
	+ p_{\ell_1 \ell_2} (\mu) p_{\ell_2 \ell_3} (\mu) g(\ell_3, mu) + ... \\
	+ p_{\ell_1 \ell_2} (\mu) p_{\ell_2 \ell_3} (\mu) ... p_{\ell_p \ell_1} (\mu) (g(\ell_p, mu) + J_{\mu}(\ell_1))
	\end{multline}
	Let $\mathbb{P}_\ell$ be the probability of receiving a ride request anywhere along the cycle, i.e. $\mathbb{P}_\ell = 1 - p_{\ell_1 \ell_2} (\mu) p_{\ell_2 \ell_3} (\mu) ... p_{\ell_p \ell_1} (\mu)$. Then 
	\begin{multline}
		\mathbb{P}_\ell J_{\mu}(\ell_1) = g(\ell_1, mu) + p_{\ell_1 \ell_2} (\mu) g(\ell_2, mu) \\
		+ p_{\ell_1 \ell_2} (\mu) p_{\ell_2 \ell_3} (\mu) g(\ell_3, mu) + ... \\
		+ p_{\ell_1 \ell_2} (\mu) p_{\ell_2 \ell_3} (\mu) ... p_{\ell_p \ell_1} (\mu) g(\ell_p, mu)
	\end{multline}
	and 
	\begin{multline}
	\mathbb{P}_\ell J_{\mu}(\ell_1) \leq p_{\ell_1 m}(\mu) V(\ell_1) + p_{\ell_1 \ell_2} (\mu) p_{\ell_2 m}(\mu) V(\ell_1) \\
	+ p_{\ell_1 \ell_2} (\mu) p_{\ell_2 \ell_3} (\mu) p_{\ell_3 m}(\mu) V(\ell_1) + ... \\
	+ p_{\ell_1 \ell_2} (\mu) p_{\ell_2 \ell_3} (\mu) ... p_{\ell_p \ell_1} (\mu) p_{\ell_p m}(\mu) V(\ell_1) \\
	= 	\mathbb{P}_\ell V(\ell_1)
	\end{multline}
	Since $\mathbb{P}_\ell > 0$, this implies that $J_{\mu}(\ell_1) \leq V(\ell_1)$.
	
	Let $\mu'$ define a policy with $\mu'(j) = \mu(j)$ for all $j \in N$ with $j \neq \ell_1$, and $\mu'_1(\ell_1) = \ell_1$ and $\mu'_2(\ell_1) = \infty$. Then as shown in \eqref{lemma_eq2}, $J_{\mu'}(\ell_1) = V(\ell_1)$. Therefore, considering also the optimality of policy $\mu$, $J_{\mu}(\ell_1) = J_{\mu'}(\ell_1)$. Due to this equality, and since $\mu'(j) = \mu(j)$ for all $j \neq \ell_1$, then $\forall i \in N$, $J_{\mu'}(i) = J_{\mu}(i)$. Therefore, $\mu'$ is an optimal stationary policy. If $\mu'$ still contains a cycle (for instance, if $\mu$ contained multiple such cycles), this procedure can be repeated until the resulting policy has no such cycles; this would require at most $n$ repetitions. 
	
	Let $\mu^*$ be the first policy constructed using (potentially multiple) iterations of the above procedure, starting from the original optimal stationary policy $\mu$, such that $\mu^*$ has no cycles. Then under $\mu^*$ there does not exist any $\{\ell_1, \ell_2, ..., \ell_p\} \subset N$ with $|\{\ell_1, \ell_2, ..., \ell_p\}| >1$, such that $\mu^*_1(\ell_1) = \ell_2$, $\mu^*_1(\ell_2) = \ell_3$, ..., $\mu^*_1(\ell_p) = \ell_1$. The set $\mu^*$ eliminates all policy cycles with two or more nodes. 

	Furthermore, in $\mu^*$, for all $i$ with $\mu^*(i) = i$, then $\mu^*_2(i) = \infty$. Thus, for any subset of nodes $\{\ell_1, \ell_2, ..., \ell_p\} \subset N$, containing an arbitrary number of nodes, $p_{\ell_1 \ell_2}(\mu^*) p_{\ell_2 \ell_3}(\mu^*) ... p_{\ell_p \ell_1}(\mu^*) = 0$. Therefore, under policy $\mu^*$, for any $i \in N$, $k, k' \in \mathbb{N}$, $\mathbb{P}(x_k = i, x_{k'} = i) = 0$. Equivalently, there exists an optimal policy $\mu^*$ that has no cycles with probability 1. \qed
\end{proof}

The policy $\mu^*$ represents an optimal stationary policy with no cycles. Therefore, every decision node is visited at most once. This implies that there are at most $n$ decision stages. Furthermore, there is an optimal stationary policy over the $n$ decision stages. Therefore, from any starting node $x_0$, we have that \eqref{DP} is exactly solved by 
\begin{equation} \label{DP2}
\max\limits_{\mu(x_k)} \mathbb{E} [\sum\limits_{k=0}^{n} g(x_k, x_{k+1}, \mu_k(x_k))]
\end{equation}
This can be solved exactly by $n$ iterations of a Value Iteration (VI) algorithm, which we present in the following section.  


\section{Path-Finding Algorithm} \label{sec:alg}
This section describes the process and the Value Iteration (VI) algorithm we developed for the between-ride problem to solve for the optimal values in equation \eqref{DP2} and, therefore, equation \eqref{DP}. First, we pre-process the network graph to ensure it satisfies the condition \eqref{assume1}. This section describes the algorithm and proves that our algorithm is optimal for the appropriately pre-processed map. 

The main algorithmic steps can be described as follows:

\begin{enumerate}
\item Initially, we calculate the expected driver revenue $\tilde{J}(x)$ for waiting at every location $x$, from \eqref{eqn:3}.
\item We ``relax" each edge iteratively to see if traveling through it will provide a more optimal path for drivers on the connecting vertices. We do this by iteratively applying \eqref{Bellman2} on every edge.
\item We terminate the algorithm when no better path is found after iterating through all edges.
\end{enumerate}

The algorithm returns a provably optimal solution. This result is proven in the next Subsection. As an added benefit, the algorithm simultaneously solves for the optimal path for all drivers in the road network. The total runtime remains the same even if more drivers are added to the network.

Figure \ref{fig:2} illustrates the first few steps of running the algorithm on a simple road network. Below, we present the pseudocode for the algorithm. We use $x.stay$ to denote $\tilde{J}(x)$ from \eqref{eqn:3}, and $e$ to denote the edge $(x,y)$.

\begin{algorithm}[H]
\caption{- Iteratively relax edges in the road network.}\label{alg1}
\begin{algorithmic}[1]
\small
\For {$x \in N$}
\State $x.value \gets x.stay$
\Comment{Initialize values to be the value for staying at that node}
\State $x.next \gets -1$
\Comment{-1 means stay at the current node}
\EndFor
\For {$i:=1$ \textbf{to} $N-1$}
\For {(x, y) $\in E$}
\State relax(x, y)
\EndFor
\EndFor
\State
\Function{relax}{x, y}
\State $x.value$ $\gets (R_e-\frac{w+fS_e}{Q_e}) (1-e^{-Q_e T_e }) +  e^{-Q_e T_e} y.value $
\If {value $> x.value$}
\State $x.next \gets y$
\State $x.value \gets$ value
\EndIf
\EndFunction
\end{algorithmic}
\end{algorithm}

\begin{algorithm}[H]
\caption{- Calculate the optimal path for the driver given its starting location.}\label{}
\begin{algorithmic}[1]
\small
\Function{getpath}{x}
\State $P \gets [x]$
\While {$x.next \neq -1$}
\State $x \gets x.next$
\State $P \gets P + [x]$
\EndWhile
\Return $P$
\EndFunction
\end{algorithmic}
\end{algorithm}

The first algorithm finds the optimal next node to go for all nodes. After processing the $x.next$ values for all the nodes, we can easily find the optimal path for any starting node using the second algorithm GETPATH, which traverses the next nodes in order.


\subsection{Proof of Correctness}

From Section \ref{results}, we have that from any starting node, there is an optimal policy that results in a path that traverses at most $n$ edges in $E$ with probability 1. This result is used to establish a proof of the optimality of the path returned by the described algorithm.

From Section \ref{results}, we see that there exists an optimal policy of \eqref{DP2} that is also optimal for \eqref{DP}. Therefore, we can focus on policies that are optimal for $n$ decision stages. These policies transverse a path of length less than or equal to $n$. This motivates the following definition:  

\begin{definition}
Let $V_i(x)$ denote the optimal value attainable if you start from $x$ with $i+1$ decision stages, i.e. you start at $x$ and travel through a path of length at most $i+1$. 
\end{definition}

Recall that $J^*(x)$ is the optimal value attainable, starting from $x$, considering all potential policies and paths of potentially infinite length. By definition of the optimal policy, $\forall x \in N$, $\forall i \in \mathbb{N}$, $V_{i}(x) \leq J^*(x)$. Proposition \ref{simplepath2} concludes that $\forall x \in N$, $V_{n}(x) = J^*(x)$. Note also that $V_0(x) = \tilde{J}(x)$ according to our initialization.
This leads to the following Proposition:

\begin{proposition}
After $n$ iterations of the \textbf{for} loops in line 4 of algorithm 1 (i.e. after we relax all the edges, repeatedly for $i$ times), for each node $x$ the stored value $V(x)$ of the node satisfies the following:\\
1. It is the value obtained from a valid policy.\\
2. It is larger than or equal to $V_i(x)$.
\end{proposition}

\begin{proof}
To prove the first property, we use induction and consider what happens when a single edge is relaxed. For the base case, at the start $V(x) = \tilde{J}(x)$, and hence $V(x)$ is the value for a valid policy (waiting at the node $x$ until a ride request is received). In the inductive case, suppose we relax edge $(x, y)$, changing the value of node $x$ from $V_{i}(x)$ to $V(x)$, with $V(x)$ as defined in line 9 of Algorithm \ref{alg1}. From the previous step, $V(y)$ is a value associated with a valid policy. It thus follows that $V(x)$ is the value of valid policy since $(x,y) \in E$.

To prove the second property, we induct on the value of $i$. In our base case $i=0$, we initialized each node with $V(x) = \tilde{J}(x) = V_0(x)$, so the second property holds immediately. Then for our inductive case, suppose value $V_i(x)$ is achieved by a valid policy $\mu$. After we relax all of the edges, we aim to show through the inductive step that $V(x) \geq V_{i+1}(x)$. Without loss of generality, assume the action associated with node $x$ under the optimal policy with $i+1$ decision stages is $\mu_{1,1}(x) = y$, i.e. the chosen direction from node $x$ in stage 1 is towards node $y$. Then, by the principal of optimality, $V_{i}(y)$ must be the value achieved by the same policy $\mu$ starting at $y$; otherwise we will be able to achieve a better value of $V_{i+1}(x)$ by selecting $\mu_{1,1}(x) = y$ and then following the policy that corresponds to $V_{i}(y)$, which violates the optimality of $V_{i+1}(x)$.

From this we know that 
$V_{i+1}(x) = (R_e-\frac{w+f S_e}{Q_e}) (1-e^{-Q_e T_e }) +  e^{-Q_e T_e} V_i(y) $
with $e = (x,y)$. By our assumption in the induction step, the value for node $y$ after $i$ iterations of the for loop must satisfy $V(y) \geq V_{i}(y)$, because our algorithm only increases the stored value $V(y)$. 

Then in the $i$-th iteration, when we update edge $(x, y)$, we have the inequality:
\begin{align*}
V(x)
& \geq (R_e-\frac{w+f S_e}{Q_e}) (1-e^{-Q_e T_e }) +  e^{-Q_e T_e} V(y) \\
& \geq (R_e-\frac{w+f S_e}{Q_e}) (1-e^{-Q_e T_e }) +  e^{-Q_e T_e} V_{i}(y) \\
& = V_{i+1}(x)
\end{align*}
which completes the proof. \qed
\end{proof}

After $n$ iterations of relaxing all edges, the values stored in each node $V(x) = V_{n}(x)$. As demonstrated in Proposition \ref{simplepath2}, $V_n(x) = J^*(x)$. 

Finally, observe that the $x.next$ values for each node provide the optimal stationary policy $\mu_1(x)$ for each node $x \in N$. Therefore, by tracing the sequence of decisions $\{x_0, x_1 = x_0.next, x_2 = x_1.next ...\}$ we can find a path that follows the optimal policy, starting from an initial node $x_0 \in N$.





The algorithm gives a worst-case runtime of $O(n |E|)$, since relaxing each edge takes constant time, and we perform $O(n |E|)$ total relaxations in lines 4 and 5 of Algorithm 1. 
We can assume a constant upper bound $C$ on the number of edges connected to any particular vertex in a transportation map, because the number of roads converging at any particular intersection has some upper bound that does not depend on the number of nodes, i.e. $|E| \leq C n$. Therefore the worst-case runtime is $O(n^2)$. \footnote{In practice, there are additional improvements that can be used to decrease the runtime. Firstly, we can terminate the algorithm if the values $V(x)$ are unchanged after a full iteration. Our runtime will then be $O(l_{max}E)$, where $l_{max}$ is the maximum length of an optimal path in the road network. This greatly reduces the runtime especially for large networks, since typically each optimal path only spans a fraction of the set of all nodes.}

\section{Implementation} \label{sec:implementation}

To demonstrate the feasibility and scalability of our algorithm, we implemented it using road network data from Boston, MA, and New York, NY. We used open source data from Open Street Maps \cite{openstreetmap} to obtain the road network data. We used posted speed limits as initial values for $S_e$ and $T_e$. We divide the map into a grid and set experimental values for $R_e$ and $Q_e$ in each gridbox. Future research could utilize TNC data to analyze the results and profit improvements for this algorithm with realistic values of $R_e$ and $Q_e$.  

\begin{figure}
\begin{center}
\includegraphics[width=1\columnwidth]{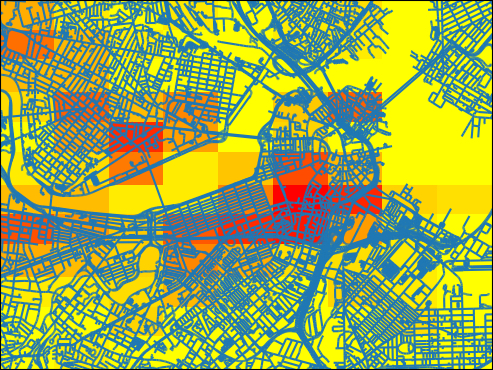}\\
\vspace{0.2cm}
\includegraphics[width=1\columnwidth]{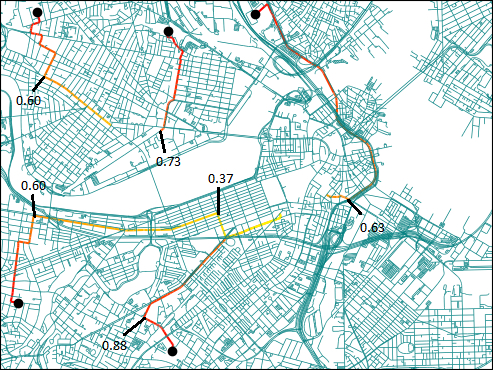}
\caption{Implementation of the algorithm in Boston, MA. The top image represents the expected value for drivers if they were to stay at that location until they get the next passenger request $\tilde{J}(x)$. Red and yellow boxes representing higher and lower values of $\tilde{J}(x)$ respectively. The bottom image represents the routes returned by optimal policies for drivers at starting nodes represented by the black dot. The color along the route represents the probability that the driver still does not have a passenger request when driving along the route at that location, with red and yellow representing higher and lower probabilities, respectively. The probability values are also marked at various points along the routes. For example, the number 0.37 is marked along one route; this implies that the driver has a 63\% chance of receiving a ride request before arriving at that location along the route. The end of the line indicates a location where the driver will optimally wait until receiving a ride request.}
\label{fig:3}
\end{center}
\end{figure}

Figure \ref{fig:3} demonstrates an example output for the city of Boston. We recorded the time taken to run the algorithm, using a standard laptop with a Intel Core i7-5500U CPU @ 2.40GHz$\times$4 processor and 8GB RAM. Overall, the program took less than 3 seconds to implement the optimization algorithm, demonstrating the feasibility of our approach. The relevant information is shown in Figure \ref{table:1}.

\begin{figure}
\small
\centering
\begin{tabular}{ |l c|  }
 \hline
 \multicolumn{2}{|c|}{City: \textbf{Boston, MA}} \\
 \hline
 Network size & 12541 Edges, 9072 Vertices\\
 Data-process time & 1.48s\\
 Algorithm time & 1.23s\\
 Total area & 41.1km$^2$\\
 Top-Right coordinates & 42.38N, 71.03W\\
 Bottom-Left coordinates & 42.33N, 71.12W\\
 \hline
 \multicolumn{2}{|c|}{City: \textbf{New York, NY}} \\
 \hline
 Network size & 20428 Edges, 13973 Vertices\\
 Data-process time & 4.25s\\
 Algorithm time & 2.31s\\
 Total area & 112.3km$^2$\\
 Top-Right coordinates & 40.82N, 73.92W\\
 Bottom-Left coordinates & 40.70N, 74.02W\\
 \hline
\end{tabular}
\caption{Details for two implementations of our algorithm. Here, the data-process time includes the time taken to parse the xml data of the road network returned by OpenStreetMaps, and store it as a graph data structure. The algorithm time includes the time taken to calculate the optimal policy for all drivers given the road network graph and applicable parameters}
\label{table:1}
\end{figure}

We can also compare the optimal driver value obtained by our approach to the value obtained by other routing decisions. For example, we can compare the optimal value to the value associated with a route where the driver takes the shortest path towards the node with the highest expected value $\tilde{J}(x)$ (we call this route the ``shortest-route"). The shortest-route represents a reasonable heuristic route for drivers in the between-ride period: head towards the highest value location. 
In the base case, the expected value of profit from the next ride, averaged across all nodes in Boston when driving along the optimal route is \$6.78, while that of the shortest-route is \$6.44. The optimal between-ride solution provides a 5\% average improvement. At certain nodes, the optimal route increases the expected value or profit by 25-50\%. 

The relative value of the optimal solution is significantly higher in periods of congestion. If we assume that average vehicle speeds are $1/2$ of the posted speed limits, then the average nodal value for the optimal policy is \$5.61, versus \$4.99 for a policy that takes the shortest path to the highest value node. In this case, the presented algorithm allows a 12\% improvement in driver profits. These calculations represent a back-of-the-envelope effort to test the value of our algorithm. Future research could use real-world price and traffic data in order to more accurately measure the benefits of our between-ride algorithm and to understand the conditions that influence its value.  


\section{Conclusion}
This paper models the between-ride routing problem for private transportation providers. We seek to optimize routing to maximize the expected value of profits for a driver that does not currently have a passenger and who is awaiting their next ride request. Our algorithm can account for various factors, including the pickup rate at different locations, surge pricing, fuel costs, and traffic conditions. 

We model the decision problem as an dynamic program with an uncertain number of stages before termination; this is equivalent to the stochastic shortest path problem. We show that under reasonable conditions, the between-ride problem can be solved to optimality by solving a simpler finite-horizon dynamic program. We present an algorithm using an iterative technique related to Value Iteration, 
and illustrate the feasibility of this algorithm by implementing it on road networks from Boston and New York City. 


There are several interesting areas for future research related to the between-ride routing problem. Our algorithm focuses on a single between-ride period. Future research could focus on the driver's optimization problem over multiple rides, considering variability in the probability distribution of ride destinations from different origins.

Our algorithm focuses on the case where the behavior of individual drivers does not substantively change the rate of passenger ride requests or the value of rides in different locations. When there are multiple drivers in the network, driver behavior could impact prices and pickup rates at different locations. In this case, the optimal between-ride behavior of drivers would anticipate the behavior and trajectory of other drivers. Extensions could use tools from game theory or mean field theory to develop optimal driver strategies that anticipate the decisions of other drivers. 
\addtolength{\textheight}{-12cm}   







\bibliographystyle{IEEEtran} 
\bibliography{IEEEexample}

\end{document}